\documentclass[superscriptaddress,twocolumn]{revtex4}
\usepackage[scaled]{helvet}

\usepackage[T1]{fontenc}

\usepackage{graphicx,amssymb,amsfonts,amsmath,amsthm,bbm}

\theoremstyle{definition}
\newtheorem{postulate}{Postulate}
\newtheorem{ppostulate}{Postulate}
\newtheorem{lemma}{Lemma}
\def\t{^{\mbox{\tiny T}}}
\def\id{\mathrm{1\hspace{-1.1mm} I}}

\begin{document}
\title{Existence of an information unit as a postulate of quantum theory}

\author{Llu\'{\i}s Masanes}
\affiliation{H.H.Wills Physics Laboratory, University of Bristol, Tyndall Avenue, Bristol BS8 1TL, U.K.}

\author{Markus P. M\"uller}
\affiliation{Perimeter Institute for Theoretical Physics, 31 Caroline Street North, Waterloo, ON N2L 2Y5, Canada}

\author{Remigiusz Augusiak}
\affiliation{ICFO-Institut de Ci\`encies Fot\`oniques, Mediterranean Technology Park, 08860 Castelldefels (Barcelona), Spain}

\author{David P\'erez-Garc\'ia}
\affiliation{Dpto. Analisis Matem\'atico and IMI, Universidad Complutense de Madrid, 28040 Madrid, Spain}

\date{\today}

\begin{abstract}
Does information play a significant role in the foundations
of physics? Information is the abstraction that allows
us to refer to the states of systems when we choose
to ignore the systems themselves. This is only possible
in very particular frameworks, like in classical or
quantum theory, or more generally, whenever there exists
an information unit such that the state of any system
can be reversibly encoded in a sufficient number of such
units. In this work we show how the abstract formalism
of quantum theory can be deduced solely from the existence
of an  information unit with suitable properties, together
with two further natural assumptions: the continuity
and reversibility of dynamics, and the possibility of
characterizing the state of a composite system by local
measurements. This constitutes a new set of postulates
for quantum theory with a simple and direct physical
meaning, like the ones of special relativity or
thermodynamics, and it articulates a strong connection
between physics and information.
\end{abstract}

\maketitle

\section{Introduction}

Quantum theory (QT) provides the foundation on top of which
most of our physical theories and our understanding of nature
sits. This peculiarly important role contrasts with our
limited understanding of QT itself, and the lack of consensus
among physicists about what this theory is saying about
how nature works.
Particularly, the standard postulates of QT are expressed
in abstract mathematical terms involving Hilbert spaces
and operators acting on them, and lack a clear physical meaning.
In other physical theories, like special relativity or
thermodynamics, the formalism can be derived from postulates
having a direct physical meaning, often in terms of the
possibility or impossibility of certain tasks. In this work
we show that this is also possible for QT.

The importance of this goal is reflected by the long history
of research on alternative axiomatizations of QT, which goes
back to Birkhoff and von Neumann~\cite{BvN, Mackey, AlfsenShultz}.
More recently, initiated by Hardy's work~\cite{5RA}, and
influenced by the perspective of quantum information theory,
there has been a wave of contributions taking a more physical
and less mathematical approach~\cite{5RA, Daki, MM, l1, CDP}.
These reconstructions of QT constitute a big achievement
because they are based on postulates having a more physical
meaning. However some of these meanings are not very direct,
and a lot of formalism has to be introduced in order to state
them. In this work we derive finite-dimensional QT from four
postulates having a clear and direct physical meaning, which
can be stated easily and without the need of heavy formalism.
Also, contrary to~\cite{Daki} we write all our assumptions explicitly.

We introduce a postulate named \textit{Existence of an Information
Unit}, which essentially states that there is only one type of
information within the theory. Consequently, any physical process
can be simulated with a suitably programmed general purpose
simulator. Since the input and output of these simulations are
not necessarily classical, this postulate is a stronger version
of the Church-Turing-Deutsch Principle (stated in~\cite{ Deutsch}).
On the other hand, it is strictly weaker than the Subspace Axiom,
introduced in~\cite{5RA} and used in~\cite{Daki} and~\cite{MM}.
An alternative way to read this postulate is that, at some level,
the dynamics of any system is substrate-independent. Within
theories satisfying the Existence of an Information Unit one
can refer to states, dynamics and measurements abstractly,
without specifying the type of system they pertain to;
and this is exploited by quantum information scientists,
who design algorithms and protocols at an abstract level,
without considering whether they will be implemented with
light, atoms or any other type of physical substrate.

\begin{figure}
    \centering
    \includegraphics[width=0.7\columnwidth]{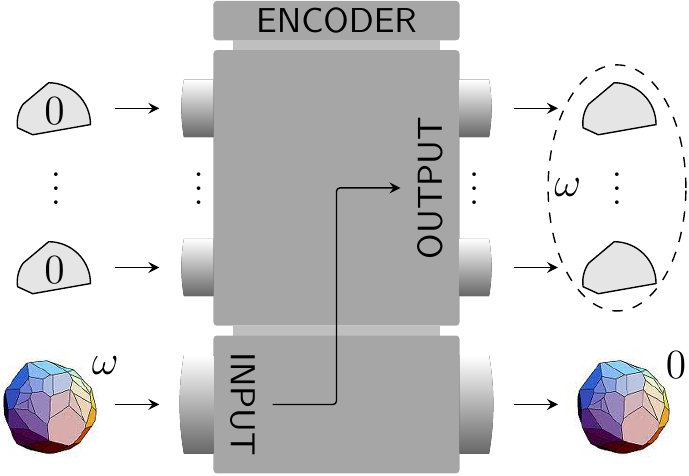}
    \caption{\textbf{Encoder.} \textit{Coding} is an ideal
    physical transformation which maps the unknown state $\omega$
    of an arbitrary system to an $n$-gbit state in a reversible
    way, and leaves the initial system in a reference state $0$.
    Reversibility means that there is another ideal physical
    transformation, \textit{decoding}, which undoes the above,
    bringing the arbitrary system back to its original state.}
    \label{encoder}
\end{figure}

More precisely, Existence of an Information Unit states that
there is a type of system, the generalized bit or \textit{gbit},
such that the state of any other system can be reversibly
encoded in a sufficient number of gbits (see Fig.~\ref{encoder}).
The reversibility of the encoding implies a correspondence
between the states of any system and the states of a multi-gbit
system (or an appropriate subspace). This correspondence also
extends to dynamics and measurements: if a given system lacks
a particular dynamics then we can encode its state into a
multi-gbit system, engineer the desired multi-gbit dynamics,
and decode back the resulting state on the given system---effectively
implementing the desired dynamics. In classical probability
theory the gbit is the bit, and in QT it is the qubit; but
we do not restrict ourselves to these two cases. We postulate
that, at some level, everything reduces to information, but
we do not specify what information is, except for some
requirements that the gbit must satisfy. One of this
requirements is No Simultaneous Encoding, which tells
that if a gbit is used to perfectly encode one classical
bit, it cannot simultaneously encode any further information.
Two close variants of this are Zeilinger's Principle~\cite{Zeilinger}
and Information Causality~\cite{IC}.

Our main contribution is to prove that QT is the only theory satisfying the
postulates of Continuous Reversibility, Tomographic Locality (both introduced
in~\cite{5RA}), The Existence of an Information Unit and No Simultaneous
Encoding. In order to prove this we make use of the classification of state
spaces performed in~\cite{GdlT,MMPA}, which shows that quantum state spaces have
very special properties. In relation to other work, in~\cite{IC} it was
suggested that Information Causality might be one of the foundational properties
of Nature. But our results support that its close variant, No Simultaneous
Encoding, might be a better candidate, since it seems to unveil more about the
structure of the physical world. Also, our results confirm Zeilinger's
idea~\cite{Zeilinger} that the limited amount of information carried by a qubit
is a defining property of QT.

\section{A theory independent formalism}

In classical probability theory, no matter how complex a system is, there is a
joint probability distribution which simultaneously describes the statistics of
all the measurements that can be performed on a system. In other words, there
exists a maximally informative measurement, of which all other measurements are
functions. This is not true in QT, and motivated by this, Birkhoff and von
Neumann generalized the formalism of classical probability theory to include
incompatible measurements~\cite{BvN}. This is nowadays called the framework of
generalized probability theories (GPTs), or the convex operational framework.

Recently, a lot of interest has been directed to the study of GPTs
\cite{5RA,Daki,MM,l1,CDP,IC,GdlT,purification,Barrett,boxworld,entropy,WO,
tradeoff,entropy2,PR,MAG}, with the double aim of reconstructing
QT, and exploring what lies beyond.
This, in particular, led to the discovery that many features originally thought
as specific to QT (such as for instance: Bell-inequality violation \cite{PR},
no-cloning \cite{Barrett,MAG}, monogamy of correlations \cite{MAG},
Heisenberg-type uncertainty relations \cite{WO,MAG}, measurement-disturbance
tradeoffs \cite{Barrett}, and the possibility of secret key distribution
\cite{BHK,PA}), are common to most GPTs.
In this light, the standard question \lq\lq{}why does nature seem to be quantum
instead of classical?\rq\rq{} sounds less appropriate than asking \lq\lq{}why QT
instead of any other GPT\rq\rq{}. Here we answer this question by showing that
any GPT different from QT violates at least one of our physically meaningful
postulates. In what follows we derive the formalism of GPTs from the basic
notions of state and measurement (a more detailed introduction can be found in
Appendix B).

In QT states are represented by density matrices. But, how can we
represent states in theories that we do not yet know? Let us follow \cite{5RA}.
The state of a system is represented by the
probabilities of some reference measurement outcomes $x_1,
\ldots x_k$ which are called \textit{fiducial}:
\begin{equation}\label{state1}
    \omega=
    \left[ \begin{array}{c}
        p(x_1) \\ \vdots \\ p(x_k)
    \end{array} \right]
    \in \mathcal{S} \subset \mathbb{R}^k\ .
\end{equation}
This list of probabilities has to be minimal but contain
sufficient information to predict the probability distribution of
all measurements that can be in principle performed on the system.
(Note that this is always possible since the list could contain
the probabilities corresponding to all measurements. In
particular, the list can be infinite, that is $k=\infty$.). The
number of fiducial outcomes $k$ is equal to the dimension of
${\cal S}$, as otherwise one fiducial probability would be
functionally related to the others, and the list not minimal. We
include the possibility that the system is present with certain
probability $U \in [0,1]$, which by consistency, is equal to the sum of
probabilities for all
the outcomes of a measurement. When the system is absent ($U=0$)
the fiducial outcomes have zero probability, hence the
corresponding state~(\ref{state1}) is the null vector ${\bf 0} \in
\mathcal{S}$. The subset of normalized states $\mathcal{N} = \{ \omega
\in\mathcal{S} : U(\omega)=1 \}$ has
dimension $k-1$.

By the rules of probability, the set of all the allowed states
$\mathcal{S}$ is convex. Indeed, by preparing the state $\omega_1$
with probability $q$ and $\omega_2$ with probability $1-q$, we
effectively prepare the mixed state $q\omega_1 + (1-q)\omega_2$.
The {\it pure states} of $\mathcal{S}$ are the normalized states
that cannot be written as mixtures. As an instance, the fiducial
outcomes for a qubit can be chosen to be $\sigma_x =1, \sigma_y
=1, \sigma_z =1, \sigma_z =-1$, and $U(\omega) = p(\sigma_z =1)
+p(\sigma_z =-1)$. Note that the set of fiducial outcomes need not be unique,
nor simultaneously measurable.

In the formalism of GPTs every convex set can be seen as
the state space $\mathcal{S}$ of an imaginary type of system, which in turn,
allows for constructing multipartite states spaces which violate Bell
inequalities more (or less) than QT.
This illustrates the degree to which this formalism generalizes classical
probability theory and QT, and allows us to catch a glimpse on the multitude of
alternative theories that we are considering here.

The probability of the measurement outcome $x$ when the system is
in the state $\omega$ is given by $E_x (\omega)$, where
\mbox{$E_x: \mathbb{R}^k \to \mathbb{R}$} is a linear function
satisfying $E_x(\mathcal{S}) \subseteq [0,1]$. To see this,
suppose the system is prepared in the mixture \mbox{$q\omega_1
+(1-q) \omega_2$}. Then the relative frequency of an outcome $x$
should not depend on whether the label of the actual preparation
$\omega_k$ is ignored before or after the measurement. As a result
\begin{equation*}
    E_x \big( q\omega_1 +(1-q) \omega_2 \big) =
    q E_x (\omega_1) +(1-q) E_x (\omega_2)\ ,
\end{equation*}
which together with $E_x ({\bf 0}) =0$ imply the linearity
of $E_x$.

Physical systems evolve with time. Often, the dynamics of a system
can be controlled by adjusting its environment, allowing in this
way to engineer different transformations of the system. A
transformation can be represented by a map $T :{\cal S} \to {\cal
S}$ which, for the same reason as outcome probabilities $E$, has
to be linear. Sometimes there are pairs of transformations whose
composition leaves the system unaffected, independently of its
initial state|in this case we say that these transformations are
reversible. The set of reversible transformations generated by
time-continuous dynamics forms a compact connected Lie group
${\cal G}$. Then, the elements of the corresponding Lie algebra
are the Hamiltonians of the theory (which in general have nothing
to do with Hermitian matrices). Our first postulate imposes that
this set of Hamiltonians is sufficiently rich.

\section{The new postulates for QT}

Now we are ready to present our new
axiomatization of QT (see Appendix A for extra discussion on the postulates).
The first postulate is motivated by the fact that most physical theories that we
know (like for example: classical mechanics, general relativity and QT) enjoy
time-continuous reversible dynamics.
\begin{postulate}[\textbf{Continuous Reversibility}]
    In any system, for every pair of pure states one can in principle engineer a
time-continuous reversible dynamics which brings one state to
the other.
\end{postulate}
Note that this postulate contains two independent assumptions: reversibility and
continuity. As pointed out by Hardy~\cite{5RA}, classical probability theory in
finite dimensions violates the continuity part of this postulate, since the set
of reversible transformations is the group of permutations, which is not
connected. Then, if we relax this continuity part, the family of theories
satisfying our postulates includes classical probability, but we do not know if
it also includes other non-classical and non-quantum theories.

Now we motivate the second postulate. Let $A$ and $B$ be two systems with
fiducial outcomes $x_1,\ldots x_{k_A}$ and $y_1,\ldots y_{k_B}$, respectively.
Is there any relation between these and the fiducial outcomes of the composite
system $AB$? The following postulate implies that the set of joint outcomes
$(x_i, y_j)$ for all $i,j$ is a fiducial set for the composite system. As a
consequence, joint local probabilities (and similarly joint local
transformations) can be obtained through the simple tensor-product rule $p(x,y)=
(E_x \otimes E_y) (\omega_{AB})$, where
\begin{equation*}
    \omega_{AB}=
    \left[ \begin{array}{c}
        p(x_1, y_1) \\ p(x_1, y_2) \\ \vdots \\ p(x_{k_A}, y_{k_B})
    \end{array} \right]
    \in {\cal S}_{AB} \subset \mathbb{R}^{k_A} \otimes \mathbb{R}^{k_B}\ .
\end{equation*}
This also implies the multiplicativity of dimensions: $k_{AB} = k_A
k_B$.
\begin{postulate}[\textbf{Tomographic Locality}]
    The state of a composite system is completely characterized by the
correlations of measurements on the individual components.
\end{postulate}
The third postulate, introduced for the first time in this work, states the
aforementioned existence of the gbit and imposes three properties that it must
satisfy.
\begin{postulate}[\textbf{Existence of an Information Unit}]
    There is a type of system (the gbit, with state space denoted ${\cal S}_{\rm
gbit}$) such that the state of any system can be reversibly encoded in a
sufficiently large number of gbits. Additionally, gbits satisfy the following:
\begin{enumerate}

    \item \textit{State Tomography Is Possible:} the state of a gbit can be
characterized with a finite number of measurements.

    \item \textit{All Effects Are Observable:} all linear functions $E: {\cal
S}_{\rm gbit} \to [0,1]$ correspond to outcomes of measurements that can in
principle be performed.

    \item \textit{Gbits Can Interact:} the group of time-continuous reversible
transformations for two gbits contains at least one element which is not product
$G_{AB} \neq G_A \otimes G_B$.
\end{enumerate}
\end{postulate}
Now, let us explain in more detail the content of Postulate~3. First, the
requirement that the state of any system can be reversibly encoded in a number
of gbits is formalized as follows.
For any state space $\mathcal{S}$ allowed by the theory there is a number $n$, a
physical transformation $T$ mapping ${\cal S}$ to the state space of $n$ gbits
$\mathcal{S}_{\rm gbit}^n$ (as in Fig.~\ref{encoder}), and another physical
transformation in the opposite direction $F: \mathcal{S}_{\rm gbit}^n \to
\mathcal{S}$, such that their composition is equal to the identity
transformation: $F(T(\omega)) =\omega$ for all $\omega \in {\cal S}$. This
implies that the dimension of $\mathcal{S}_{\rm gbit}^n$ is not smaller than
that of $\mathcal{S}$. If the two dimensions are equal then the two state spaces
are equivalent. But if the dimension of $\mathcal{S}_{\rm gbit}^n$ is larger
than that of $\mathcal{S}$ then there are states in $\mathcal{S}_{\rm gbit}^n$
which are not contained in $T(\mathcal{S})$; and for those the transformation
$F$ does not work with unit probability. Next, we explain the properties that
gbits satisfy.

\begin{enumerate}

\item The fact that gbits can be characterized with a finite number of
measurements is equivalent to say that the dimension of the state space
$\mathcal{S}_{\rm gbit}$, denoted $k_{\rm gbit}$, is finite. This may seem
contradictory with the fact that in quantum theory, there is a type of
tomography for infinite-dimensional systems. But these systems have an infinite
number of perfectly distinguishable states, hence, after imposing additional
constrains (like an upper bound on the energy) the effective Hilbert space is
finite, and state tomography becomes possible. However, as a consequence of No
Simultaneous Encoding, gbits have only two perfectly distinguishable states.

\item In classical probability theory and QT, all effects correspond to outcomes
of measurements. This need not be the case in general, but in order to single
out QT, we have to impose it on gbits. Although in this form this assumption
does not have a direct operational meaning, it can be formulated in a way that
it does (see~\cite{CDP} or Appendix A.5). Unfortunately, this alternative
formulation is more cumbersome, hence we avoid it here.

\item Interaction is fundamentally necessary in order not to have an essentially
trivial universe. The requirement that any system can be reversibly encoded in
gbits implies that, if gbits do not interact among them, then no other system
interacts. Postulate~3.3 rules out this possibility.
\end{enumerate}

\begin{postulate}[\textbf{No Simultaneous Encoding}] If a gbit is used to
perfectly encode one classical bit, it cannot simultaneously encode any further
information.
\end{postulate}

To illustrate Postulate~4 let us consider a communication task involving two
distant parties, Alice and Bob.
Similarly as in the scenario for Information Causality~\cite{IC}, suppose that
Alice is given two bits $a,a' \in \{0,1\}$, and Bob is asked to guess one of
them. He will base his guess on information sent to him by Alice, encoded in one
gbit. Alice encodes the gbit with no knowledge of which of the two bits, $a$ or
$a'$, Bob will try to guess.
No Simultaneous Encoding imposes that, in a coding/decoding strategy in which
Bob can guess $a$ with probability one, he knows nothing about $a'$. That is, if
$b,b'$ are Bob's guesses for $a,a'$ then
\[
	P(b|a,a')= \delta^a_b \ \ \Rightarrow\ \
	P(b'|a,a'=0)= P(b'|a,a'=1)
\]
where $\delta^a_b$ is the Kronecker tensor.
A straightforward consequence of this is that $\mathcal{S}_\mathrm{gbit}$
contains at most two perfectly distinguishable states. Other consequences are
derived below.

Another way to state No Simultaneous Encoding is: suppose that Alice encodes
$a,a\rq{}$ in the four states $\omega_{a,a\rq{}} \in \mathcal{N}_\mathrm{gbit}$.
If there is an effect $E$ such that $E(\omega_{a,a\rq{}}) = \delta_{a,0}$ then
any effect $E\rq{}$ satisfies $E\rq{} (\omega_{a,0}) = E\rq{} (\omega_{a,1})$.
As it is illustrated in Fig.~\ref{fwic}, this together with All Effects Are
Observable (cf. Postulate~3.2) imply that all states in the boundary of
$\mathcal{N}_\mathrm{gbit}$ are pure (first arrow in Figure~3).

\begin{figure}\label{fwic}
    \centering
    \includegraphics[height=34mm]{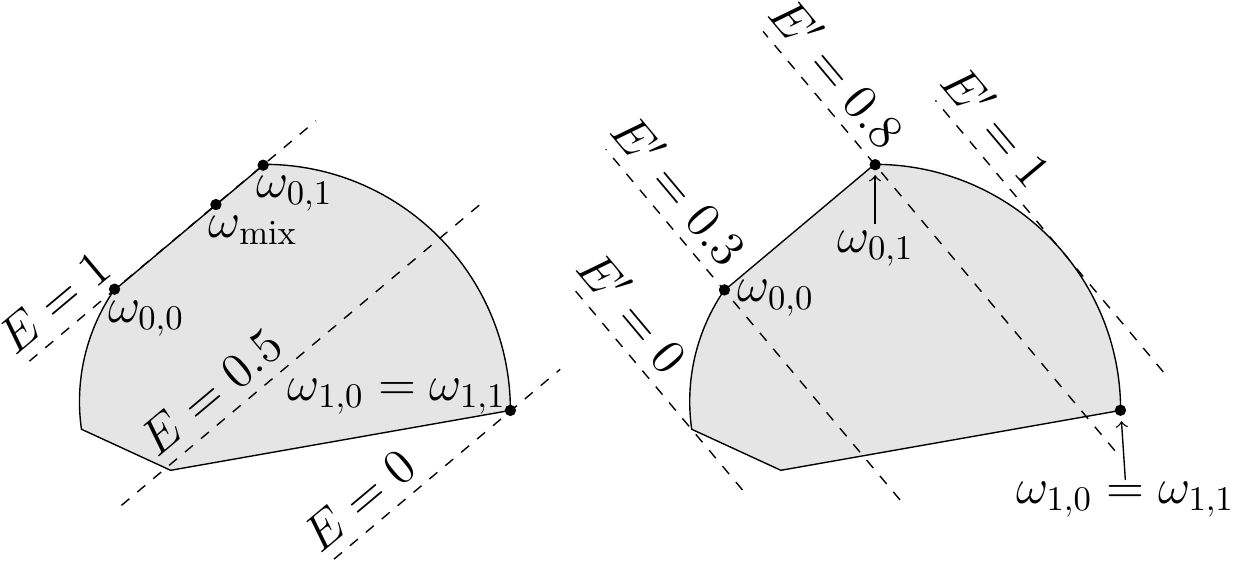}
    \caption{\textbf{No Simultaneous Encoding.} This figure shows that there
cannot be mixed states in the boundary of $\mathcal{N} _\mathrm{gbit}$. If there
is one, say $\omega_{\mathrm{mix}}$, then this boundary contains a non-trivial
face (left figure). Since all effects are observable, we can decode $a$ with the
effect $E$, which gives probability one for all states inside that facet, and
probability zero for some other state(s). By encoding $(a,a\rq{})= (0,0), (0,1)$
in two different states inside that face we can perfectly retrieve $a$ through
$E$, while still getting some partial information about $a\rq{}$ with another
effect $E\rq{}$ (right figure).}
\end{figure}

An interesting remark is that our four postulates, except for part 2 of
Postulate 3, express the possibility or impossibility of certain tasks. This is
very similar in spirit to formulations of the second law of thermodynamics, the
principle of equivalence of gravitation and inertia, or the principle of light
speed invariance. Contrary, this remains completely hidden in the standard
postulates of QT.

\section{Argumentation}

Having stated our four postulates, let us now show that the only
theory obeying them is QT. In what follows we present an overview
of the proof, while its detailed version can be found in Appendix D. First of
all, Postulate 3.1 implies that the dimension of the gbit $k_\mathrm{gbit}$ is
finite. Then, Continuous Reversibility associates to
any state space $\mathcal{S}$ a group of reversible
transformations $\mathcal{G}$, having an invariant scalar product
with respect to which all pure states of $\mathcal{S}$ have the
same norm. This together with the fact that the boundary of
$\mathcal{N}_\mathrm{gbit}$ contains only pure states imply that
it is an ellipsoid (second arrow in Figure~\ref{argumentation}).
By setting as the new set of fiducial outcomes the effects
corresponding to the principal axes of the ellipsoid (recall that
all effects are observable), $\mathcal{N}_\mathrm{gbit}$ becomes a
Euclidean ball (third arrow in Figure~\ref{argumentation}). But
what is the state space of two gbits
$\mathcal{S}_\mathrm{gbit}^2$?

\begin{figure}\label{argumentation}
    \centering
    \includegraphics[height=70mm]{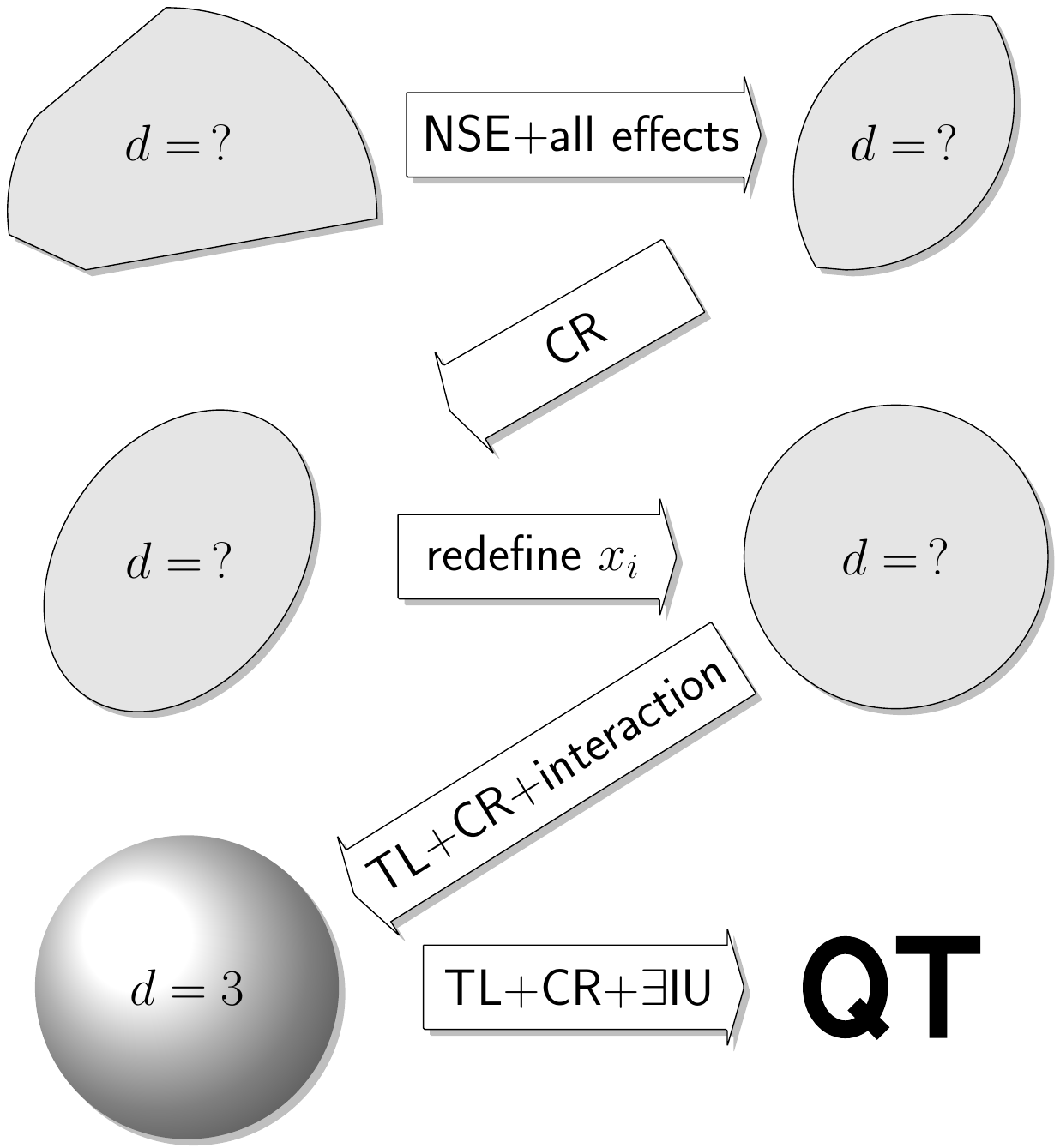}
    \caption{\textbf{Summary of the argumentation.} This figure synthesizes the
proof that the only theory satisfying our four postulates is QT. Each step
(represented by an arrow) invokes part of the content of the postulates
(specified inside the arrow) and reveals new information about the state space
of the generalized bit. Initially (top-left) $\mathcal{N} _\mathrm{gbit}$ is an
arbitrary convex set with arbitrary dimension $d= k_\mathrm{gbit} -1$, and
finally (down-left) it is a 3-dimensional ball. The first arrow represents the
step explained in Figure~2. The abbreviations CR, TL, $\exists$IU, NSE, ``all
effects'' and ``interaction'' refer respectively to Continuous Reversibility,
Tomographic Locality and Existence of an Information Unit, No Simultaneous
Encoding, All Effects Are Observable, Gbits Can Interact.}
\end{figure}

According to Continuous
Reversibility the set of pure states of two gbits can be written
as $\{ G(\omega \otimes \omega) | G \in
\mathcal{G}_\mathrm{gbit}^2 \}$, where
$\mathcal{G}_\mathrm{gbit}^2$ is the group of reversible
transformations for two gbits, and $\omega$ is a pure state of one
gbit. The group $\mathcal{G}_\mathrm{gbit}^2$ is unknown, but by
consistency, it must contain all local transformations
\begin{equation}\label{GG in G2}
  \mathcal{G}_\mathrm{gbit} \otimes \mathcal{G}_\mathrm{gbit}
  \ \subseteq\  \mathcal{G}_\mathrm{gbit}^2\ ,
\end{equation}
and
it must generate states with well-defined probabilities, meaning
that
\begin{equation}\label{G constraint}
    (E_x \otimes E_y) (G (\omega \otimes \omega)) \in [0,1]\ ,
\end{equation}
holds for all $G \in \mathcal{G}_\mathrm{gbit}^2$ and any (local) gbit
effects $E_x, E_y$. The family of all bipartite state spaces satisfying these
two consistency requirements was analyzed in~\cite{MMPA}, and it was shown that,
with the exception of the quantum case, all state spaces contain separable
states only, and the corresponding groups $\mathcal{G}_\mathrm{gbit}^2$ contain
product transformations only. But this is in contradiction with Gbits Can
Interact! Hence, the combination of this postulate together with
requirements~\eqref{GG in G2} and~\eqref{G constraint} is very restrictive, and
it implies that the Euclidean ball $\mathcal{N}_\mathrm{gbit}$ has dimension
$k_\mathrm{gbit} -1=3$ and $\mathcal{G}_\mathrm{gbit} = \mathrm{SO} (3)$ (see
Appendix D and~\cite{MMPA}). This tells us that, locally, gbits are identical to
qubits, but it is not clear yet whether
multi-gbit state spaces $\mathcal{S}_\mathrm{gbit}^n$ having a
non-quantum structure are consistent with our postulates. In
Reference~\cite{GdlT} all possible joint state spaces of $n$
systems that are locally qubits are classified, and it is found
that the only possibility allowing for non-product reversible
transformations is multi-qubit QT. So gbits must be locally and
globally like qubits: $\mathcal{S} _\mathrm{gbit}^n$ is the set of
$2^n$-dimensional density matrices and
$\mathcal{G}_\mathrm{gbit}^n$ is the adjoint representation of
$\mathrm{SU} (2^n)$. Finally, since any state space is reversibly
encodable in a multi-qubit system, the states, transformations and
measurements of any system can be represented within the formalism
of finite-dimensional QT.

\section{Conclusions}

Given the controversy around the foundations of QT, it is very natural to seek
for modifications and generalizations of QT. And some authors claim that this is
necessary in order to unify the description of quantum and gravitational
phenomena~\cite{Penrose, HG}.
Each set of postulates for QT provides a different starting point for this
endeavor. For example, starting from the standard postulates, some authors have
modified the Schr\"odinger equation~\cite{Weinberg}, or the field of numbers
over which the Hilbert space is defined~\cite{quaternions}. But a radically
different starting point is provided by our postulates.
In Appendix D we relax that Gbits Can Interact (Postulate~3.3) and characterize
the family of theories that emerges (see also~\cite{MMPA}). It is shown that all
these alternative theories, though being not classical, do not contain
entanglement and do not violate Bell inequalities. If instead, we relax the
continuity part of the Continuous Reversibility Postulate, then the family of
theories that emerges includes classical probability theory, but we leave for
future research whether other theories are included as well. This seems an
important question, because in our construction and others~\cite{5RA}, the
continuity of the dynamics appears to be the dividing feature between classical
probability theory and QT.

A repeated pattern in the history of science is the promotion of a scientific
instrument to a model for understanding the world. For
instance, there are some proposals for viewing the universe as a giant computer
(classical~\cite{Zuse} or quantum~\cite{Lloyd}).
But what is the physical content of this? Can the dynamics of any system be
understood as computation? After all it is computing its
future state. We propose that a requisite for upgrading time-evolution to
computation is that such time-evolution is substrate-independent, in the sense
that it can be simulated in a system of information units. In this work we have
taken this perspective seriously: we have promoted the Existence of an
Information Unit with suitable properties to be a postulate, and we have shown
that this together with the very natural postulates of Continuous Reversibility
and Tomographic Locality, uniquely determine the full mathematical formalism of
QT.

\section*{Acknowledgments}
Ll. M. acknowledges support from CatalunyaCaixa, the EU ERC Advanced Grant NLST, the EU Qessence project, the Templeton Foundation and the FQXi large grant project ``Time and the structure of quantum theory''. Research at Perimeter Institute for Theoretical Physics is supported in part by the Government of Canada through NSERC and by the Province of Ontario through MRI. R. A. acknowledges support from AQUTE, TOQATA, and Spanish MINCIN through the Juan de la Cierva program. D. P.-G. acknowledges support from the Spanish grants MTM2008-01366 and S2009/ESP-1594.

\newpage
\appendix
\section*{Appendices}

The following appendices contain some remarks on the postulates for quantum
theory (QT) that we have presented, a thorough introduction to the formalism of
generalized probability theory (GPT), and a rigorous proof of the claims made in
the article.

\section{Some remarks on the postulates}

\subsection{Continuous Reversibility}

The postulate of Continuous Reversibility was introduced
in~\cite{5RA}, under the name of \lq\lq{}continuity axiom\rq\rq{}.
One of the motivations to assume the reversibility and continuity
of time evolution is that the most fundamental theories that we
know, classical or quantum, enjoy it. The meaning of continuity
here is that, when the system evolves for a very small time, the
initial and the final states are almost indistinguishable. This is
equivalent to the connectedness of the group of dynamical transformations.

Up to present-day experimental accuracy, time evolution seems to
be continuous. But it is conceivable that at a small scale it is
discrete, and continuity is only an approximation that is valid at
sufficiently large scales. In this case, our postulates could be
understood as describing the corresponding large-scale effective
theory.

A very interesting open problem is the classification of
theories which satisfy all our postulates except for the
continuity part of Postulate~1, that is, when the group of
reversible transformations $\mathcal{G}$ is not required to be
connected. One theory of this kind is classical probability
theory, but it is not known if there are others. In~\cite{MM} it is shown that,
if in addition one assumes the postulate of \lq\lq{}Equivalence of
Subspaces\rq\rq{}, which is arguably very strong, the only theories that survive
are QT and classical probability theory.

\subsection{Tomographic Locality}

The axiom of Tomographic Locality has a direct operational
meaning, but additionally, it is mathematically very natural,
since it endows state spaces of multipartite systems with the
familiar tensor-product structure. The authors of~\cite{HW}
consider ways of relaxing Tomographic Locality.

\subsection{Existence of an Information Unit}

Any state of a quantum system can be encoded with arbitrary
precision in a sufficient number of classical bits. For instance,
this can be achieved by writing its density matrix in a bit
string. However, if we are given a quantum system in an unknown
state, there is no way we can obtain this bit string, unless we
are given a large number of copies of the system. By measuring a
single copy of the system we could encode the outcome in a bit
string, but there is no way we can prepare the same state if the
only information we have is this bit string. In other words, this
encoding is not reversible.

In summary, the classical bit does not constitute a unit of
information capable of reversibly encoding the state of any
quantum system, although it does if we restrict to classical
systems. However, according to QT, the qubit does constitute such
a unit of information, and we think that this is a fundamental
aspect of QT. Hence, in this work we promote this to postulate.

Our approach can be summarized in the following slogan:
\textit{Information does play a significant role in the
foundations of physics, but we do not say what information
actually is.} In this sense, our postulates specify some
properties that information must satisfy, but they do not right
away specify its physical implementation. That is, they do not
postulate that information must be quantum---instead, this fact is
\textit{derived as a consequence} of the properties that
information should satisfy.

\subsection{State Tomography Is Possible}

In finite-dimensional quantum system, the dimension $k$ and the number of
perfectly distinguishable states $c$ are related through the
equation $k= c^2$. (Do not confuse the dimension of the set of unnormalized
density matrices $k$ with the dimension of the associated Hilbert space $c$.)
However, for arbitrary state spaces, the only constraint between
the positive integers $k$ and $c$ is $k \geq c$. Hence, although
not very natural, it is possible that systems with only two
perfectly distinguishable states (like, for instance, gbits) have
infinite dimension. However, since for any finite $k_\mathrm{gbit}
\neq 3$ interaction between gbits is impossible, we are inclined to
think that $k_\mathrm{gbit} =\infty$ is also incompatible with
Postulate~3.3. (In order to prove this, transitive groups on the
infinite-dimensional euclidean sphere should be considered.) Consequently, we
conjecture that Postulate~3.1 is redundant, but, since we cannot prove this
fact, we keep the
postulate.

Independently of the above discussion, the finiteness of
$k_\mathrm{gbit}$ is necessary if we want state tomography to be
possible. The fact that in QT state tomography of
infinite-dimensional systems is possible is due to the fact that
these systems also have an infinite number of perfectly
distinguishable states, and with a bound on the energy, one can
effectively consider the system to be finite-dimensional. But this
does not work if the infinite-dimensional system has two distinguishable states,
like a gbit. Additionally, it is desirable to perform tomography on the unit of
information with no need of extra assumptions, like upper bounds on the energy.

\subsection{All Effects Are Observable}
\label{SubsecAllEffects}

A priori, given any state space of a physical system, all effects (i.e., linear
functionals that yield valid probabilities between $0$ and $1$ on all states)
describe outcome probabilities of conceivable measurements.
However, one might imagine that there are additional physical
restrictions, similar to the superselection rules, that somehow
render some of the effects impossible to appear in actual
measurements. Our postulate says that we do not consider this more
complicated situation: we assume that, at least in
principle, every effect can appear as the outcome of some
measurement.

An interesting fact is that this postulate can be weakened without affecting the
conclusions of our work. Instead of all the effects, only effects $E$ for which
there are two states $\omega_0, \omega_1 \in \mathcal{S} _\mathrm{gbit}$ such
that $E(\omega_0)=0$ and $E(\omega_1)=1$ need to be observable. This second
statement is logically equivalent to
Chiribella-D\rq{}Ariano-Perinotti\rq{}s information-theoretic postulate named
\lq\lq{}Perfect Distinguishability\rq\rq{} (see~\cite{CDP}), phrased as
\lq\lq{}every state that is not
completely mixed can be perfectly distinguished from some other state\rq\rq{}.
In their notation, a state $\omega \in \mathcal{S}$ is \lq\lq{}completely
mixed\rq\rq{} if for any other state $\omega_1 \in \mathcal{S}$ there is a
decomposition of $\omega$ of the form $\omega= p\omega_1 + (1-p)\omega_2$ with
$p>0$. Which can be interpreted as the state $\omega$ being compatible with the
preparation of any other state $\omega_1$. Our choice of Postulate~3.2 is
motivated by simplicity.

\subsection{Gbits Can Interact}

Interaction is necessary for the creation of entanglement, and
consequently, for the violation of Bell inequalities. But even
more, without interaction classical computation is impossible,
since single-gbit gates cannot be universal. More generally, the
emergence of structure and complex systems seems impossible in a
world without interaction. For these reasons we find it very
natural to postulate that gbits can interact.

We claim that one can explore what lies beyond QT by relaxing some
of our postulates. For example, the family of theories which are
compatible with all our postulates except for Gbits Can Interact (Postulate~3.3)
is given in~\cite{MMPA}.
Obviously, in all these other theories there is no
entanglement.

\subsection{No Simultaneous Encoding}

Let us describe the two differences between No Simultaneous Encoding (NSE) and
Information Causality (IC). First, the communication task associated to IC can
be seen as a teleportation analog of the one associated to NSE. That is, in IC,
Alice sends Bob classical information in a context where Bell-violating
correlations are shared, while in NSE, Alice sends Bob a possibly non-classical
system in a context where no correlations are shared. Note that these two
communication tasks become equivalent in theories where steering or
teleportation are possible.

Second, the trade-off between Bob\rq{}s knowledge on $a$ and $a\rq{}$ imposed by
NSE is based on the guessing probability, while the one of IC is based on the
Shannon mutual information, which in this context lacks an operational meaning.
Additionally, the bound based on the guessing probability is weaker than the one
based on the mutual information.

On its own, Information Causality suffices to derive Tsirelson\rq{}s
bound~\cite{T} and other constraints on quantum correlations~\cite{ABPS}, but
not all of them~\cite{GWAN}. Here we show that the full structure of quantum
correlations can be derived from a variant of Information Causality together
with our other postulates. This gives an answer to the question of how to
characterize all quantum correlations from physical principles.

\section{Introduction to generalized probability theories}\label{gpt}

In this section we introduce a formalism that allows us to
represent states, measurements and transformations in a
theory-independent way. More complete material can be found
in~\cite{l1,Pfister}.

\subsection{States}\label{ss}

In this formalism, the state of a system is represented by the
probabilities of some reference measurement outcomes $x_1,
\ldots x_k$ which are called \textit{fiducial}:
\begin{equation}\label{state}
    \omega=
    \left[ \begin{array}{c}
        p(x_1) \\ \vdots \\ p(x_k)
    \end{array} \right]
    \in {\cal S} \subset \mathbb{R}^k\ .
\end{equation}
This list of probabilities has to be minimal but contain
sufficient information to predict the probability distribution of
all measurements that can be in principle performed on the system.
Note that this is always possible since the list could contain the
probabilities corresponding to all measurements. In particular,
the list can be infinite, i.e.\ $k=\infty$. We include the
possibility that the system is present with a certain probability
$u \in [0,1]$. This probability is given by the \textit{unit effect}, $u=
U(\omega)$, which is equal to the sum of
probabilities for all the outcomes of a measurement. When the
system is absent ($u=0$) the fiducial outcomes have zero
probability, hence the corresponding state~(\ref{state}) is the
null vector ${\bf 0} \in \mathcal{S}$. The subset of normalized
states is $\mathcal{N} = \{\omega\in\mathcal{S}|U(\omega)=1\}$.
Clearly, any state $\omega \in \mathcal{S}$ can be written as $\omega = u \nu$,
where $u= U(\omega)$ is the norm of $\omega$, and $\nu \in \mathcal{N}$ is the
normalized version of $\omega$. This last statement is equivalent to the fact
that $\mathcal{S}$ is the convex hull of $\mathcal{N}$ and ${\bf 0}$
(see~\cite{convex_book} for a definition of convex hull).

By the rules of probability, the set of all the allowed states
$\mathcal{S}$ is convex. Indeed, by preparing the state $\omega_1$
with probability $q$ and the state $\omega_2$ with probability
$1-q$, we effectively prepare the mixed state $q\omega_1 +
(1-q)\omega_2$. The {\it pure states} of $\mathcal{S}$ are the
normalized states that cannot be written as mixtures, that is, the
extremal points of $\mathcal{N}$. Hence, we denote the set of pure
states by $\mathrm{ext} \mathcal{N}$. The number of fiducial
outcomes $k$ is equal to the dimension of ${\cal S}$, as otherwise
one fiducial probability would be functionally related to the
others, and the list not minimal. Hence, the dimension of
$\mathcal{N}$ is $k-1$. As an instance, the fiducial outcomes for
a quantum two-level system (or qubit) can be chosen to be
$\sigma_x =1, \sigma_y =1, \sigma_z =1, \sigma_z =-1$; hence,
$k=4$ and $U(\omega) = p(\sigma_z =1) +p(\sigma_z =-1)$. Note,
however, that the set of fiducial outcomes need not be unique, nor
simultaneously measurable. The role of fiducial outcomes is comparable to that of basis vectors in linear algebra.

By changing the set of fiducial outcomes one can transform the
geometry of a state space. However, as shown in the next
paragraph, all such transformations are linear and invertible.
Conversely, all invertible linear transformations generate an
equivalent state space, hence, state spaces are equivalence
classes of convex sets under linear equivalence. Indeed, for any
invertible linear transformation $L: \mathbb{R}^k \to
\mathbb{R}^k$, we can redefine the states $\omega \to L (\omega)$
and the effects $E \to E\circ L^{-1}$ such that the physics is
unchanged $(E\circ L^{-1}) (L (\omega)) = E(\omega)$. In a similar
fashion, by redefining the transformations as $T \to L\circ T
\circ L^{-1}$, the dynamical structure of the system is unchanged
$(L\circ T \circ L^{-1})(L(\omega))= L(T(\omega))$. Hence, every
possible state space is an equivalence class of convex sets
related by linear transformations. Note that in general, the
components of the vector $L(\omega)$ are not in $[0,1]$, so we
cannot interpret them as fiducial probabilities. However, as
illustrated below, sometimes it is advantageous to loose the
probability interpretation of the components of $L(\omega)$ in
favor of a different representation that is easier to handle.

\subsection{Measurements}

The probability of the measurement outcome $x$ when the system is
in state $\omega$ is given by $E_x (\omega)$ where \mbox{$E_x:
\mathbb{R}^k \to \mathbb{R}$} is a linear functional satisfying
$E_x(\mathcal{S}) \subseteq [0,1]$. To see this, suppose the
system is prepared in the mixture \mbox{$q\omega_1 +(1-q)
\omega_2$}. Then the relative frequency of an outcome $x$ should
not depend on whether the label of the actual preparation
$\omega_k$ is ignored before or after the measurement. As a result
\begin{equation*}
    E_x \big( q\omega_1 +(1-q) \omega_2 \big) =
    q E_x (\omega_1) +(1-q) E_x (\omega_2),
\end{equation*}
which together with $E_x ({\bf 0}) =0$ imply the linearity of
$E_x$. Linear functions $E$ satisfying $E(\mathcal{S}) \subseteq
[0,1]$ are called \textit{effects} and can be written as a scalar
product $E(\omega)=E\cdot\omega=\sum_{i=1}^k E^i p(x_i)$ with $E$
being a vector from $\mathbbm{R}^k$. An effect that plays a
special role is the unit effect $U(\omega)= \sum_{i=1}^k U^i p(x_i)$, which
gives the probability that the system is present. In classical
probability theory and QT, all effects correspond to outcomes of
measurements, but this need not be the case in general (this is
related to the discussion in Subsection~\ref{SubsecAllEffects}).
Below we postulate this to hold for gbits.

An $n$-outcome measurement is represented by $n$ effects $E_1, \ldots, E_n$
satisfying
\begin{equation*}
    E_1 + \cdots + E_n =U\ .
\end{equation*}
Alternatively speaking, this formula means that the outcome
probabilities are normalized, implying that we only need to
specify $n-1$ effects. In particular, a two-outcome measurement is
represented by a single effect $E$, which, for a normalized state
$\omega\in \mathcal{N}$, gives outcome probabilities $E(\omega)$
and $1-E(\omega)$. We say that $\omega_1, \ldots, \omega_n \in
\mathcal{S}$ are perfectly distinguishable states if there is an
$n$-outcome measurement in $\mathcal{S}$ such that $E_i (\omega_j)
= \delta_{ij}$, where $\delta_{ij}$ is the Kronecker tensor.

\subsection{Transformations}\label{s trans}

Physical systems evolve with time. Often, the dynamics of a system
can be controlled by adjusting its environment, allowing in this
way to engineer different transformations of the system. A
transformation is represented by a map $T : \mathcal{S} \to
\mathcal{S}$ which, for the same reason as outcome probabilities
$E$, has a linear extension $T :\mathbb{R}^k \to
\mathbb{R}^k$ and satisfies the consistency constraint
$T(\mathcal{S}) \subseteq \mathcal{S}$. Using linearity and the decomposition
$\omega = u \nu$ (with $u= U(\omega)$ and $\nu \in \mathcal{N}$) we have
$(U\circ T)(\omega) = u (U\circ T) (\nu) \leq u = U(\omega)$ for all $\omega \in
\mathcal{S}$. We write this inequality with the short-hand notation
\begin{equation}\label{non-increasing}
	U\circ T \preceq U\ ,
\end{equation}
meaning that it holds for all states of the corresponding state space.
Equation~(\ref{non-increasing}) is a generalization of the quantum requirement
that physical operations (apart from being completely positive maps) must not
increase the trace.

Sometimes there are pairs of transformations $T,F: \mathcal{S} \to \mathcal{S}$
whose composition leaves the system unaffected, independently of its initial
state: $T\circ F=I$, where $I$ is the identity transformation. Using elementary matrix theory we know that, if this holds then the equality $F\circ T=I$ holds too; hence we say that both transformations are \textit{reversible}; and we write $T^{-1} =F$ and $F^{-1} =T$. Note that the invertibility of the matrix $T$ associated to a physical transformation does not imply that its inverse $T^{-1}$ satisfies the consistency constraints of a physical transformation $T^{-1} (\mathcal{S}) \subseteq \mathcal{S}$, and that it is allowed by the theory. Hence, reversibility is more restrictive than invertibility. If $T$ is reversible then inequality~(\ref{non-increasing}) gives $U = (U\circ T^{-1}) \circ T \preceq U\circ T \preceq U$, which implies $U\circ T=U$. The physical interpretation of this last equality is: \textit{reversible transformations are deterministic}.

The set of transformations generated by time-continuous reversible
dynamics forms a connected matrix group ${\cal G}$. From a
physical point of view, it makes sense to include in $\mathcal{G}$
all transformations which can be approximated arbitrarily well by
those allowed by the theory, or equivalently, we assume that $\mathcal{G}$ is topologically closed. Therefore, $\mathcal{G}$ is a compact matrix group, which according to~\cite{matrix_groups}, must be a Lie group. The elements of the corresponding Lie algebra are the Hamiltonians of the theory (which in general have nothing to do with Hermitian matrices; even in QT, these would be ``superoperators'' acting on the space of density matrices). The postulate of Continuous Reversibility imposes that this set of Hamiltonians is sufficiently rich.

One can implement transformations which, in addition to a possible
change of state, also transform the type of system. A
transformation that takes a system from a state space
$\mathcal{S}_1$ and outputs a system from a different state
space $\mathcal{S}_2$, with respective dimensions $k_1$ and $k_2$,
can be represented by a linear map $T: \mathbb{R}^{k_1} \to
\mathbb{R}^{k_2}$ satisfying the consistency constraint
$T(\mathcal{S}_1) \subseteq \mathcal{S}_2$.
If a physical theory forbids the ``transmutation'' of types of systems then
transformations which effectively modify the type of system can still be
implemented with the method described in Fig.~1, where the input is an
$\mathcal{S}_1$-system in an arbitrary state $\omega \in \mathcal{S}_1$ together
with an $\mathcal{S}_2$-system in a fixed state 0, and the output is an
$\mathcal{S}_1$-system in a fixed state 0 together with an
$\mathcal{S}_2$-system in the output
state $T(\omega)\in \mathcal{S}_2$. As in~(\ref{non-increasing}), one can show
that $U_2 \circ T \preceq U_1$, where $U_1, U_2$ are the unit effects of
$\mathcal{S}_1, \mathcal{S}_2$ respectively. And again,  the equality $U_2 \circ
T = U_1$ holds when the transformation $T$ is deterministic.

We say that a transformation which modifies the type of system $T:
\mathcal{S}_1 \to \mathcal{S}_2$ is reversible if there is another
transformation which modifies the type of system in the opposite direction $F:
\mathcal{S}_2 \to \mathcal{S}_1$ such that $F\circ T = I_1$, where $I_1$ is the
identity transformation on $\mathcal{S}_1$. The following lemma establishes some
properties of this type of transformation.

\begin{lemma}\label{l1}
If the linear maps $T: \mathbb{R}^{k_1} \to \mathbb{R}^{k_2}$ and $F:
\mathbb{R}^{k_2} \to \mathbb{R}^{k_1}$ satisfy the constraints $T(\mathcal{S}_1)
\subseteq \mathcal{S}_2$, $F(\mathcal{S}_2) \subseteq \mathcal{S}_1$ and $F\circ
T = I_1$ then
\begin{enumerate}
	\item $k_1 \leq k_2$,
	\item $T$ is deterministic,
	\item $F$ succeeds with unit probability when the input is any state
$\omega \in T(\mathcal{S}_1) \subseteq \mathcal{S}_2$,
	\item $(T\circ F)(\omega) = I_2 (\omega)$ for all $\omega \in
T(\mathcal{S}_1) \subseteq \mathcal{S}_2$,
	\item If $k_1 = k_2$ then the two state spaces are equivalent
$T(\mathcal{S}_1)= \mathcal{S}_2$.

\end{enumerate}
\end{lemma}

\begin{proof}
Suppose that $k_1 > k_2$. Then the matrix $F\circ T$ is not full-rank, and
cannot be equal to the identity. Therefore $k_1 \leq k_2$. As shown above, the
premises of the lemma imply that $U_2 \circ T \preceq U_1$ and $U_1 \circ F
\preceq U_2$. This together with $I_1 = F\circ T$ gives $U_1 = U_1 \circ F \circ
T \preceq U_2 \circ T \preceq U_1$, which implies $U_2 \circ T= U_1$, or in
other words: $T$ is deterministic.
Another consequence of the reversibility premise $F\circ T = I_1$ is that when
restricted to the subset $\omega\in T(\mathcal{S}_1) \subseteq \mathcal{S}_2$ we
have $T\circ F |_{T(\mathcal{S}_1)}  = I_2 |_{T(\mathcal{S}_1)}$, which for the
same reason as above, it implies  $U_1 \circ F |_{T(\mathcal{S}_1)} = U_2
|_{T(\mathcal{S}_1)}$, or in other words: the transformation $F$ is
deterministic when restricted to states $\omega \in T(\mathcal{S}_1) \subseteq
\mathcal{S}_2$.

Finally, let us consider the case $k_1 = k_2$. The square matrices $T$ and $F$
are respective inverses ($T=F^{-1}$ and $F= T^{-1}$) and hence bijective. This
together with $T(\mathcal{S}_1) \subseteq \mathcal{S}_2$ and $T^{-1}
(\mathcal{S}_2) \subseteq \mathcal{S}_1$ implies that $T(\mathcal{S}_1) =
\mathcal{S}_2$ and $F(\mathcal{S}_2) = \mathcal{S}_1$.
\end{proof}

\subsection{Composite systems}\label{cs}

To a setup as the one appearing in Fig.~\ref{f1} we associate a
system if, for each configuration of the preparation,
transformation, and measurement devices, the relative frequencies
of the outcomes tend to a unique probability distribution. Two
systems $A,B$ constitute a composite system $AB$ if a measurement
for $A$ together with a measurement for $B$ uniquely specifies a
measurement for $AB$, independently of the temporal ordering. The
fact that subsystems are systems themselves implies that each
global state $\omega_{AB}$ has well-defined reduced states
$\omega_{A}, \omega_B$ which do not depend on which
transformations and measurements are performed on the other
subsystem; this is often referred to as no-signaling. Some
bipartite correlations satisfying the no-signaling constraint
violate Bell inequalities more than QT does~\cite{PR}; however, as
we will show, these are incompatible with our postulates.

\begin{figure}
    \centering
    \includegraphics[height=3cm]{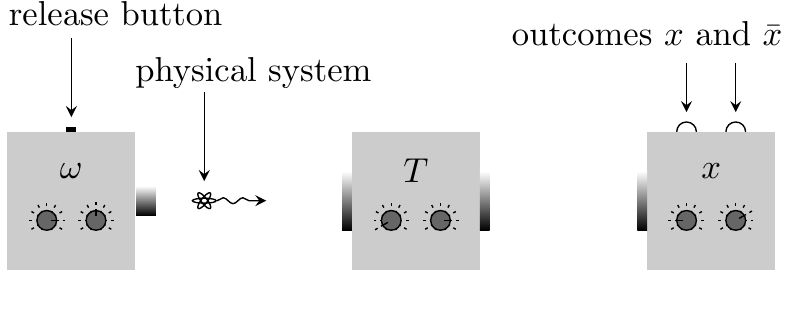}
    \caption{{\bf General experimental setup.} From left to right there are
the preparation, transformation and measurement devices. As soon
as the release button is pressed, the preparation device outputs a
physical system in the state specified by its knobs. The next
device performs the transformation specified by its knobs (which
in particular can be ``do nothing''). The device on the right
performs the measurement specified by its knobs, and the outcome
\mbox{($x$ or $\bar{x}$)} is indicated by the corresponding
light.}
    \label{f1}
\end{figure}

A bipartite system is also a system, so its states can be
represented by the probabilities of some fiducial outcomes. But
what is the relationship between these and the fiducial outcomes
of the subsystems, $x_1, \ldots, x_{k_A}$ and $y_1, \ldots,
y_{k_B}$? In order to answer this question, we point out that the
fact that $p(x,y)$ does not depend on the ordering of the
measurements giving outcomes $x,y$ implies the following

\begin{lemma}\label{lemma2}
The joint probability $p(x,y)$ of any pair of subsystem outcomes $x,y$ is given
by
\begin{equation}\label{prob rule}
    p(x,y)= (E_x \otimes E_y) \cdot\omega_{AB}\ ,
\end{equation}
where
\begin{equation}\label{psiAB}
    \omega_{AB} = \left[ \begin{array}{c}
    p(x_1, y_1)\\ p(x_1, y_2)\\ \vdots\\ p(x_{k_A}, y_{k_B})
    \end{array} \right]\ \in\
    \mathbb{R}^{k_A} \otimes \mathbb{R}^{k_B}\ .
\end{equation}
Product states and the set of all these vectors $\omega_{AB}$ span
the vector space $\mathbb{R}^{k_A} \otimes \mathbb{R}^{k_B}$.
\end{lemma}

\begin{proof}
If the system $B$ is measured first, giving outcome $y_j$, then
the system $A$ is in the state determined by the fiducial
probabilities $p(x_i|y_j)= p(x_i, y_j)/p(y_j)$, and the
single-system probability rule can be applied $p(x|y_j)= \sum_i
E^i_x\, p(x_i|y_j)$. Multiplying by $p(y_j)/ p(x)$ and using
Bayes' rule gives
\[
   p(y_j |x)= \sum_i E^i_x\, p(x_i,y_j) /p(x).
\]
By using the freedom in the ordering of measurements, we can
interpret $p(y_j |x)$ as the state of the system $B$ once the
system $A$ has been measured giving outcome $x$, and the
single-system probability rule can be applied again: $p(y_j|x)=
\sum_j E^j_y\, p(y_j |x)= \sum_{i,j} E^i_x\, E^j_y\, p(x_i,y_j)
/p(x)$. Multiplying both sides of this equality by $p(x)$ gives
(\ref{prob rule}).

Let us see that the vectors $\omega_{AB} \in \mathcal{S}_{AB}$ span the full
tensor product space. In QT, the only states $\omega_{AB} \in \mathcal{S}_{AB}$
which
have pure states as marginals $\omega_A\in \mathcal{S}_A ,
\omega_B \in \mathcal{S}_B$, are product ones $\omega_{AB}=
\omega_A \otimes \omega_B$. The same proof technique applies to
generalized probability theories.
This implies that
$\mathcal{S}_{AB}$ contains all product states, otherwise there
would be a state in $\mathcal{S}_A$ or $\mathcal{S}_B$ which is
not the marginal of any state in $\mathcal{S}_{AB}$.
Next, note that by minimality, $\mathcal{S}_A$ contains $k_A$
linearly independent vectors, and analogously for $\mathcal{S}_B$.
The tensor products of these vectors are a set of $k_{AB}= k_A
k_B$ linearly independent vectors in $\mathcal{S}_{AB}$, so the
set $\mathcal{S}_{AB}$ has full dimension.
\end{proof}

And what about global measurements? The postulate of Tomographic
Locality states that the probability for the outcome of any
measurement, local or global, is determined by the joint
probability $p(x,y)$ of all local measurements. This implies that
$\omega_{AB}$ in~(\ref{psiAB}) constitutes a complete
representation of a bipartite state, since all outcome
probabilities can be calculated from it. Hence, the linear span of
$\mathcal{S}_{AB}$ is $\mathbb{R}^{k_A k_B}$, which implies that
dimensions follow a multiplicative rule.
\begin{equation}\label{pr}
    k_{AB} = k_A k_B\ .
\end{equation}
From now on, we use this tensor-product representation given by
Eqs. (\ref{prob rule}) and (\ref{psiAB}) for bipartite states. In
this representation, the marginal states are given by $\omega_A =
(\id \otimes U) (\omega_{AB})$ and $\omega_B = (U \otimes \id)
(\omega_{AB})$. For a given pair of states spaces
$\mathcal{S}_{A}, \mathcal{S}_{B}$ the composite state space
$\mathcal{S}_{AB}$ is not unique in general.
The only consistency constraints on $\mathcal{S}_{AB}$ are:
\begin{enumerate}
    \item $\mathcal{S}_{AB}$ must contain the set of separable states, that is
the convex hull of $\mathcal{S}_{A} \otimes \mathcal{S}_{B}$,

    \item all states $\omega\in \mathcal{S}_{AB}$ must give valid probabilities
$(E_x \otimes E_y)(\omega) \in [0,1]$ for all local measurements $x,y$.
\end{enumerate}

\section{Statement of the postulates}

Next, using the formalism developed in the previous section, we rewrite our
postulates in a more formal way. But before that, let us introduce some
notation. We denote by ${\cal S}_{\rm gbit}$ the (yet unknown) state space of a
gbit, and by $k_{\rm gbit}$ its corresponding dimension. Also, we denote by
${\cal S}_{\rm gbit}^n$ the state space of $n$ gbits, and $k^{(n)}_{\rm gbit}$
its corresponding dimension.

\begin{ppostulate}[\textbf{Continuous Reversibility}]
    In any system, the group of  transformations $\mathcal{G}$
 generated by time-continuous reversible dynamics is transitive on
the set of pure states $\mathrm{ext} \mathcal{N}$.
\end{ppostulate}

\begin{ppostulate}[\textbf{Tomographic Locality}]
    The state of a composite system is completely
characterized by the correlations of measurements on the
individual components: $p(x,y)$ for all local outcomes $x,y$.
\end{ppostulate}

\begin{ppostulate}[\textbf{Existence of an Information Unit}]
    There is a type of system, \textit{the gbit},
which satisfies the following:
\begin{enumerate}
\setcounter{enumi}{-1}
    \item For each state space $\mathcal{S}$ there is a
number $n$ and two physical transformations, $T: \mathcal{S} \to {\cal S}_{\rm
gbit}^n$ and $F: {\cal S}_{\rm gbit}^n \to \mathcal{S}$, such that $T\circ F =
I$.

  \item The state space of a gbit ${\cal S}_\mathrm{gbit}$ has finite dimension
$k_\mathrm{gbit}$.

    \item All effects on a gbit correspond to measurement outcomes.

       \item The group of transformations generated by time-continuous
reversible
dynamics of two gbits ${\cal G}_{\rm gbit}^2$ contains an element
which is not product ($G_{AB} \neq G_A \otimes G_B$).
\end{enumerate}
\end{ppostulate}

\begin{ppostulate}[\textbf{No Simultaneous Encoding}] If there are four gbit
states $\omega_{a,a\rq{}} \in \mathcal{S} _{\rm gbit}$ (with $a,a\rq{} \in
\{0,1\}$) and an effect $E$ such that $E(\omega_{a,a\rq{}})= \delta_{a,0}$, then
any effect $E\rq{}$ satisfies $E\rq{} (\omega_{a,0}) = E\rq{} (\omega_{a,1})$.
\end{ppostulate}

We have seen that Postulate~2 implies equation~(\ref{pr}), hence $k^{(n)}_{\rm
gbit} = k_{\rm gbit}^n$. Postulate~3.0 provides the premises of Lemma~\ref{l1},
then the results of this lemma follow too. In particular, the number of gbits
$n$ has to be sufficiently large for $k^n_{\rm gbit} \geq k$ to hold, where $k$
is the dimension of the arbitrary state space $\mathcal{S}$. But the main
consequence of Postulate~3.0 is that we only need to characterize the state
spaces ${\cal S}_{\rm gbit}^n$; and once this is done, we know that any state
space $\mathcal{S}$ compatible with our Postulates must be a subspace of ${\cal
S}_{\rm gbit}^n$ for some value of $n$. And, as already discussed, this does not
only establish a correspondence between the states of $\mathcal{S}$ and the
states in the subspace of ${\cal S}_{\rm gbit}^n$, but also a correspondence
between the measurements and the transformations which keep the subspace of
${\cal S}_{\rm gbit}^n$ invariant.

In the next Section we show that the only possible state space ${\cal
S}_\mathrm{gbit}^n$ compatible with our postulates is the set of
$2^n$-dimensional quantum density matrices
\begin{equation}\label{nqs}
	S^{2^n}_\mathrm{QT} = \left\{ \rho \in \mathbb{C}^{2^n \times 2^n} |\
\rho^\dagger =\rho,\ \mathrm{tr}\rho \leq 1,\ \rho \geq 0\ \right\}
\ ,
\end{equation}
with associated set of effects $\rho \mapsto \mathrm{tr} M\rho$, where $M$ is
any $2^n$-dimensional, complex matrix satisfying $0 \leq M \leq \id$; and group
of reversible transformations $\rho \mapsto U\rho U^\dagger$, for all $U\in
\mathrm{SU}(2^n)$. In other words, gbits are quantum two-level systems (or
qubits), and they combine into composite systems in exactly the way prescribed
by QT. Thus, our postulates single out all state spaces ${\cal S}$ that can be
simulated on $n$-qubit systems, that is, $k$-level quantum systems for some $k
\in \mathbb{N}$, and quantum systems with linear constraints on the density
matrix elements, such as classical systems or systems with superselection rules.

\section{Proof of equivalence}

\subsection{A single gbit}\label{asg}

\begin{lemma}
    Postulate~3.2 and Postulate~4 imply that there are no mixed states in the
boundary of $\mathcal{N} _\mathrm{gbit}$.
\end{lemma}

\begin{proof}
Suppose the mixed state $\omega_\mathrm{mix}= q \omega_1 + (1-q)
\omega_2$ is in the boundary of $\mathcal{S}_\mathrm{gbit}$. Then,
there exists an effect $E$ with $E(\omega_\mathrm{mix}) =1$ and
$E(\omega\rq{})=0$ for some other state $\omega\rq{} \in
\mathcal{S} _\mathrm{gbit}$. According to Postulate~3.2 this
effect is in principle measurable. Moreover, the linearity of $E$
together with the property $E(\mathcal{S} _\mathrm{gbit}) \in
[0,1]$ imply that $E(\omega_1)= E(\omega_2)= 1$. Therefore, we can
encode $a=0$ in $\omega_1$ or $\omega_2$, and $a=1$ in
$\omega\rq{}$. Additionally, we can encode $a\rq{}=0$ in
$\omega_1$ and $a\rq{}=1$ in $\omega_2$. Since $\omega_1 \neq
\omega_2$, there is an effect $E\rq{}$ for which $E\rq{}
(\omega_1) \neq E\rq{} (\omega_2)$. By relabeling $\omega_{0,0}=
\omega_1$, $\omega_{0,1}= \omega_2$ and $\omega_{1,0}=
\omega_{1,1}= \omega\rq{}$ we obtain a contradiction with
Postulate~4.
\end{proof}
Figure~2 contains a pictorial representation of the above proof.

\begin{lemma}\label{l3}
    Continuous Reversibility together with the fact that
$\mathcal{N} _\mathrm{gbit}$ has no mixed states in its boundary
imply that $\mathcal{N} _\mathrm{gbit}$ is a solid ellipsoid.
\end{lemma}

\begin{proof}
Using the Haar measure on the compact connected Lie group
$\mathcal{G} _\mathrm{gbit}$, we can define a positive matrix
\begin{equation}\label{W}
    W^2= \int_{\mathcal{G} _\mathrm{gbit}}
    \hspace{-5mm} dG\ G\t G  \ ,
\end{equation}
and $W$ as its unique positive square root. Note that $W\t =W$ and
$W^2 G^{-1} = G\t W^2$ for all $G\in \mathcal{G} _\mathrm{gbit}$.
According to Continuous Reversibility, for any pair of pure states
$\omega_1, \omega_2 \in \mathcal{S} _\mathrm{gbit}$ we have
$\omega_2 = G\omega_1$ for some $G\in \mathcal{G} _\mathrm{gbit}$,
and hence
\begin{eqnarray*}
    |W\omega_2| &=& \sqrt{\omega_2 \cdot W^2 \omega_2}=
    \sqrt{\omega_1 \cdot G\t W^2 G \omega_1}
\\
    &=& \sqrt{\omega_1 \cdot W^2 \omega_1 }= |W\omega_1|\ ,
\end{eqnarray*}
where the notation $\omega_1 \cdot \omega_2$ is used to denote the
Euclidean inner product, while, accordingly, $|\omega_1|=
\sqrt{\omega_1 \cdot \omega_1}$ stands for the Euclidean norm.

This allows us to define the constant $r= |W\omega|$, where
$\omega \in \mathcal{S} _\mathrm{gbit} $ is  a pure state. Note
that $r$ is independent of the chosen pure state $\omega$. The set
$\mathcal{E}= \{x\in\mathbb{R} ^{k_\mathrm{gbit}}|  r= |Wx|\}$ is
an ellipsoid, and the intersection of $\mathcal{E}$ and the
normalization hyperplane $\mathcal{F}= \{x\in\mathbb{R}
^{k_\mathrm{gbit}} | U\!\cdot x=1\}$ is also an ellipsoid. The
pure states of $\mathcal{N} _\mathrm{gbit}$ are contained in the
intersection $\mathcal{E} \cap \mathcal{F}$, and since there are
no mixed states in the boundary of $\mathcal{N} _\mathrm{gbit}$,
the set of pure states $\mathrm{ext} \mathcal{N} _\mathrm{gbit}$
must be $\mathcal{E} \cap \mathcal{F}$, which is a
$(k_\mathrm{gbit}-1)$-dimensional ellipsoid.
\end{proof}

\begin{lemma}[\textbf{Bloch-vector representation}]\label{l4}
Postulates~1, 3.2 and 4 imply the existence of a representation where the state
space of a gbit is
\begin{equation}
\label{sb}
    \mathcal{S}_\mathrm{gbit} =
    \left\{
    u \left[ \begin{array}{c}
        1 \\ \hat \omega
    \end{array} \right]
    |\  u\in [0,1],\ \hat \omega \in \mathbb{R}^d ,\ |\hat \omega| \leq 1
    \right\},
\end{equation}
the normalization effect is $U= [1, \hat{\mathbf{0}}]$, and the
group of transformations generated by time-continuous dynamics is
\begin{equation}
\label{gb}
    \mathcal{G}_\mathrm{gbit} =
    \left\{
    \left[ \begin{array}{cc}
        1 & \mathbf{0} \\
        \mathbf{0} & \hat G
    \end{array} \right]
    |\  \hat G \in \hat{\mathcal{G}} _\mathrm{gbit}
    \right\},
\end{equation}
where $\hat{\mathcal{G}} _\mathrm{gbit}$ is a compact connected subgroup
of $\mathrm{SO} (d)$ which is transitive in the unit sphere of
$\mathbb{R}^d$, and $d= k_\mathrm{gbit} -1 \geq 2$.
\end{lemma}

\begin{proof}
First, we follow the reparametrization procedure explained at the
end of Section~\ref{ss}. In this case, the invertible
transformation is $L= \sqrt{2}\, r^{-1} W$, where $r$ and $W$ are
defined in the proof of Lemma~\ref{l3}. All the matrices of the
reparametrized group $\tilde{\mathcal{G}} _\mathrm{gbit} = L\circ
\mathcal{G} _\mathrm{gbit} \circ L^{-1}$ are orthogonal. To see
this, note that
\begin{eqnarray*}
    \tilde G\t \tilde G &=& (WGW^{-1})\t (WGW^{-1})
\\ &=&
    \int_{\mathcal{G} _\mathrm{gbit}}
    \hspace{-5mm} dH\ W^{-1} G\t H\t H G W^{-1}
\\ &=&
    W^{-1} W^2 W^{-1} = \id\ ,
\end{eqnarray*}
where we have used the fact that $W$ is symmetric. From now on,
when referring to the state space, effects and transformations of
a gbit, we mean the reparametrized ones,
\begin{align}
    \mathcal{S} _\mathrm{gbit} &\to L (\mathcal{S} _\mathrm{gbit})\ ,
\nonumber \\
    U\t &\to U\t L^{-1}\ ,
\\ \nonumber
    \mathcal{G} _\mathrm{gbit} &\to
    L\, \mathcal{G} _\mathrm{gbit} L^{-1}\ ,
\end{align}
and we omit the tilde.

The orthogonality of the transformations in $\mathcal{G}
_\mathrm{gbit}$ implies that a left eigenvector $U\t G=G$ is also a
right eigenvector $GU=U$. Hence, the matrix group $\mathcal{G}
_\mathrm{gbit}$ contains a trivial one-dimensional representation
spanned by $U$, and another representation denoted
$\hat{\mathcal{G}} _\mathrm{gbit}$. So, for any $G\in \mathcal{G}
_\mathrm{gbit}$ there is $\hat G \in \hat{\mathcal{G}}
_\mathrm{gbit}$ such that
\begin{equation}
    G = \left[ \begin{array}{cc}
        1 & \mathbf{0} \\
        \mathbf{0} & \hat G
    \end{array} \right]\ .
\end{equation}
In our notation, symbols with a hat ``$\hat{\phantom{o}}$'' are
associated to the non-trivial representation of $\mathcal{G}
_\mathrm{gbit}$. In this basis, the normalization effect is $U=
[1, \hat{\mathbf{0}}]$, and the pure states are $\omega= [1, \hat
\omega]$ with $|\hat \omega|=1$, where the latter is a consequence
of the fact that, according to the definition of $L$, pure states
have Euclidean norm $|\omega| =\sqrt{2}$.
\end{proof}

From now on, when dealing with a gbit, we adopt the representation
given in Lemma~\ref{l4}. Note that, in this representation, pure
states are those with unit normalization $u=1$, and unit-length
Bloch vector $|\hat\omega|=1$. Each effect is characterized by a
vector $E= [e, \hat E]$ such that
\begin{equation}
    E(\omega)=  E\cdot \omega= u(e+ \hat E \cdot \hat\omega)\ .
\end{equation}
The consistency constraint $E(\mathcal{S} _\mathrm{gbit}) \subseteq [0,1]$ is
equivalent to $|\hat E|\leq e$ and $e+ |\hat E| \leq 1$. An effect $E$ for which
there are two states $\omega_0, \omega_1 \in \mathcal{S} _\mathrm{gbit}$ such
that $E(\omega_0)=0$ and $E(\omega_1)=1$, satisfies $e= |\hat E|= 1/2$. Such
effects are in one-to-one correspondence with pure states $\omega \in
\mathrm{ext} \mathcal{N} _\mathrm{gbit}$ via the map $E= \omega/2$.

\subsection{Two gbits}
\label{SubsecTwoGbits}
Tomographic Locality and Lemma~\ref{l4} imply that two-gbit states can be
represented as
\begin{equation}\label{b2}
    \omega_{AB} =
    u \left[ \begin{array}{c}
            1  \\
            \alpha \\
            \beta \\
            \gamma
    \end{array} \right] \in \mathcal{S} _\mathrm{gbit} ^2\ ,
\end{equation}
where $\alpha= \hat{\omega}_A \in \mathbb{R}^d$, $\beta=
\hat{\omega}_B \in \mathbb{R}^d$, and $\gamma \in \mathbb{R}^d
\otimes \mathbb{R}^d$ is called the ``correlation matrix''. Note
that the ordering of the components in~(\ref{b2}) is different
from the one in~(\ref{psiAB}). At this stage, we know that
$|\alpha|, |\beta| \leq 1$, but we do not know much about the full
structure of $\mathcal{S} _\mathrm{gbit} ^2$, nor its associated
group $\mathcal{G} _\mathrm{gbit} ^2$. However, these two objects
are very much related. Indeed, the postulate of Continuous
Reversibility implies that the set of pure states for two gbits is
$\mathrm{ext} \mathcal{N} _\mathrm{gbit} ^2 = \mathcal{G}
_\mathrm{gbit} ^2 (\omega \otimes \omega)$, where $\omega \in
\mathrm{ext} \mathcal{N} _\mathrm{gbit}$ is any pure state. In
order to see this, recall that product states belong to
$\mathcal{S} _\mathrm{gbit} ^2$, and that the product of two
locally pure states is a globally pure state. This connection
between $\mathcal{S} _\mathrm{gbit} ^2$ and $\mathcal{G}
_\mathrm{gbit} ^2$ implies that the consistency constraints for
$\mathcal{S} _\mathrm{gbit} ^2$ mentioned at the end of
Section~\ref{cs}, translate to constraints for $\mathcal{G}
_\mathrm{gbit} ^2$. These constrains are the premise of
Lemma~\ref{l5}.

\begin{lemma}\label{l5}
    Let $\hat{\mathcal{G}}_\mathrm{gbit}$ be a connected subgroup of
$\mathrm{SO} (d)$
 which is transitive in the unit sphere of $\mathbb{R} ^d$, where
$d\geq 2$. Let $\mathcal{G} _\mathrm{gbit} ^2$ be a connected
group of real $(d+1)^2 \times (d+1)^2$ matrices which satisfies
the following:
\begin{enumerate}
    \item $(\mathcal{G} _\mathrm{gbit} \otimes \mathcal{G} _\mathrm{gbit}) \leq
\mathcal{G} _\mathrm{gbit} ^2$,
    \item $(E\otimes E) \cdot G (\omega \otimes \omega) \in [0,1]$ for all $G
\in \mathcal{G} _\mathrm{gbit} ^2$,
\end{enumerate}
where $\omega = [1, \hat\omega]$, $|\hat\omega|=1$ and $E=
\omega/2$. If $d\neq 3$, then the group $\mathcal{G}
_\mathrm{gbit} ^2$ must be a subgroup of $\mathcal{H}_d \otimes
\mathcal{H}_d$, where
\begin{equation}
\label{qb}
    \mathcal{H}_d =
    \left\{
    \left[ \begin{array}{cc}
        1 & \mathbf{0} \\
        \mathbf{0} & Q
    \end{array} \right]
    |\  Q \in \mathrm{SO} (d)
    \right\}\ .
\end{equation}
\end{lemma}

\begin{proof}
See Reference~\cite{MMPA}.
\end{proof}

Lemma~\ref{l5} shows that, except for $d=3$, the joint state space
$\mathcal{S} _\mathrm{gbit} ^2$ only contains separable states,
and its associated group $\mathcal{G} _\mathrm{gbit} ^2$ only
contains non-interacting transformations, which is in
contradiction with Postulate~3.3. Hence, the only possibility is
$d=3$, and this is getting quite close to QT. It turns out that
the only subgroup of $\mathrm{SO} (3)$ which is transitive on the
unit sphere of $\mathbb{R} ^3$ is $\mathrm{SO} (3)$ itself. Hence,
from now on, we assume $d=3$ and $\hat{\mathcal{G}} _\mathrm{gbit}
= \mathrm{SO} (3)$.

\subsection{Emergence of quantum theory}

Since QT satisfies our postulates, it must fit the structure that
we have found up to this stage. Indeed, the state of a qubit can
be represented by the three-dimensional Bloch vector $\hat\omega$
(and the normalization parameter $u$ if we consider general
states). That is, the state space $\mathcal{S}_\mathrm{gbit}$ of a
gbit and that of a qubit, $\mathcal{S}_\mathrm{qubit}$, are
equivalent in the sense of Subsection~\ref{ss}: both are in
one-to-one correspondence via an invertible linear map $L$. It is
given by
\begin{eqnarray*}
   L:\mathcal{S}_\mathrm{gbit}&\to&\mathcal{S}_\mathrm{qubit} \\
   \left[\begin{array}{c} 1 \\ \hat\omega\end{array}\right] &\mapsto& \frac 1
2\left(1\cdot\id +\hat\omega\cdot\vec\sigma\right),
\end{eqnarray*}
where $\vec\sigma=(\sigma_x,\sigma_y,\sigma_y)$ is a vector with
the Pauli matrices as entries. This maps Bloch vectors onto
density matrices,
$\mathcal{S}_\mathrm{qubit}=\{\rho\in\mathbb{C}^{2\times
2}\,\,|\,\, \rho\geq 0, {\rm tr}(\rho)=1\}$. The group of
reversible transformations for one gbit is $\hat{\mathcal{G}}
_\mathrm{gbit} =\mathrm{SO} (3)$, which is equivalent to the
adjoint representation of $\mathrm{SU} (2)$:
\[
   \mathcal{G}_\mathrm{qubit} = L \mathcal{G}_\mathrm{gbit} L^{-1} = \left\{
\rho\mapsto U \rho U^\dagger\,\,|\,\, U \in SU(2)\right\}.
\]
The group of reversible transformations for $n$ qubits, denoted $\mathcal{G}
_\mathrm{qubit}^n$, is the adjoint action of $\mathrm{SU} (2^n)$:
\[
   \mathcal{G}_\mathrm{qubit}^n  = \left\{ \rho\mapsto U \rho U^\dagger\,\,|\,\,
U \in SU(2^n)\right\}.
\]
As shown in~\cite{MMPA}, this is the \textit{only} possible choice
of any $n$-gbit dynamics which satisfies our postulates---up to an
equivalence transformation, which is $L$ for a single gbit
(mapping Bloch vectors to density matrices), and, correspondingly,
$L^{\otimes n}=L\otimes\ldots\otimes L$ for $n$ gbits, mapping the
corresponding state vectors to density matrices of size $2^n$:
\begin{lemma}\label{l6}
    Let $\hat{\mathcal{G}}_\mathrm{gbit} = \mathrm{SO} (3)$ and let $\mathcal{G}
_\mathrm{gbit} ^n$ be a connected group of $(3+1)^n \times (3+1)^n$ real
matrices which satisfies the following:
\begin{enumerate}
    \item $(\mathcal{G} _\mathrm{gbit} \otimes \cdots \otimes \mathcal{G}
_\mathrm{gbit}) \leq \mathcal{G} _\mathrm{gbit} ^n$,
    \item $(E\otimes \cdots \otimes E) \cdot G (\omega \otimes\cdots \otimes
\omega) \in [0,1]$ for all $G \in \mathcal{G} _\mathrm{gbit} ^n$,
\end{enumerate}
where $\omega = [1, \hat\omega]$, $|\hat\omega|=1$ and $E= \omega/2$. The only
possible groups $\mathcal{G} _\mathrm{gbit} ^n$ are:
\begin{enumerate}
    \item $\mathcal{G} _\mathrm{gbit} \otimes \cdots \otimes \mathcal{G}
_\mathrm{gbit}$

    \item $\left(L^{-1}\right)^{\otimes n} \mathcal{G}_\mathrm{qubit}^n
L^{\otimes n}$.
\end{enumerate}
\end{lemma}

\begin{proof}
See Reference~\cite{MMPA}.
\end{proof}

Let us summarize. In Section~\ref{asg} we proved that the set of normalized
states of a gbit must be an Euclidean unit ball. In Section~\ref{SubsecTwoGbits}
we have shown that this ball must have dimension three. Hence, the state space
of a single gbit is identical to that of a single qubit. Note that this does not
in itself automatically imply that the state space of $n$ gbits is identical to
the state space of $n$ qubits; however, the previous lemma shows that it does by
invoking our postulates.
Hence, all systems allowed by our postulates can be described within the quantum
formalism.


\end{document}